\documentclass[english,15pt]{article}
\usepackage{hyperref}
\usepackage[T1]{fontenc}
\usepackage[latin1]{inputenc}
\usepackage{geometry}
\usepackage{float}
\usepackage{mathrsfs}
\usepackage{amsmath}
\usepackage{amssymb}
\usepackage{graphicx}

\hypersetup{
    colorlinks=true,
    linkcolor=blue,
    filecolor=magenta,      
    urlcolor=cyan,
    citecolor = blue,
}

\makeatletter

\usepackage{algorithm,algpseudocode}

\usepackage{amsthm}

\usepackage{mathrsfs}

\usepackage{amsfonts}

\usepackage{epsfig}

\usepackage{bm}

\usepackage{mathrsfs}

\usepackage{enumerate}

\numberwithin{equation}{section}

\@ifundefined{definecolor}{\@ifundefined{definecolor}
 {\@ifundefined{definecolor}
 {\usepackage{color}}{}
}{}
}{}

\usepackage{subfig}\usepackage[all]{xy}

\newtheorem{lem}{Lemma}[section]

\newtheorem{prop}{Proposition}[section]

\newcounter{hypA}

\newcounter{hypB}

\newcounter{hypD}

\usepackage{babel}\date{}

\usepackage{babel}

\makeatother

\usepackage{babel}

\begin{document}

\begin{center}

{\Large \textbf{Bayesian Parameter Inference for Partially Observed Stochastic Volterra Equations}}

\vspace{0.5cm}

AJAY JASRA, HAMZA RUZAYQAT \& AMIN WU.

{\footnotesize 
Applied Mathematics \& Computational Science Program, \\
Computer, Electrical and Mathematical Sciences and Engineering Division, \\ King Abdullah University of Science and Technology, Thuwal, 23955, KSA.} \\
{\footnotesize E-Mail:\,} \texttt{{\footnotesize  ajay.workj.jasra100@gmail.com, \\ hamza.ruzayqat@kaust.edu.sa, amin.wu@kaust.edu.sa}} \\

\end{center}

\begin{abstract}
In this article we consider Bayesian parameter inference for a type of partially observed stochastic Volterra equation (SVE).
SVEs are found in many areas such as physics and mathematical finance. In the latter field they can be used
to represent long memory in unobserved volatility processes. In many cases of practical interest, SVEs must be time-discretized and then
parameter inference is based upon the posterior associated to this time-discretized process. Based upon recent studies on time-discretization
of SVEs (e.g.~\cite{richard}) we use Euler-Maruyama methods for the afore-mentioned discretization. We then show how multilevel Markov chain Monte
Carlo (MCMC) methods \cite{jasra} can be applied in this context. In the examples we study, we give a proof that shows that the cost to achieve a mean square error (MSE) of $\mathcal{O}(\epsilon^2)$, $\epsilon>0$, is {$\mathcal{O}(\epsilon^{-\tfrac{4}{2H+1}})$, where $H$ is the Hurst parameter. If one uses a single level MCMC method then the cost is 
$\mathcal{O}(\epsilon^{-\tfrac{2(2H+3)}{2H+1}})$} to achieve the same MSE. We illustrate these results in the context of state-space and stochastic volatility models, with the latter applied to real data.\\

\end{abstract}

\section{Introduction}
\label{sec:intro}

Stochastic Volterra equations (SVEs) are natural extension of the well-known deterministic Volterra equations and have been introduced in physics \cite{grip} and later on adapted in mathematical finance \cite{euch}. {In finance, these models describe the price process of a financial asset and are successfully used in equity markets to price
derivative options. These models take into account the effect of past returns
on the volatility. They are also capable of reproducing the correct form of the implied volatility skew at small times \cite{gatheral,euch}.} In this article we investigate the scenario that the SVE is partially observed, via discrete (in time) data and with unknown parameters. Given access to a fixed data set and a prior on the unknown parameters, we wish to estimate the parameters from the associated posterior distribution. This has several real applications, such as in econometrics via stochastic volatility models; see e.g.~\cite{bergomi, jasra_levy}. The longer memory that can be present in stochastic volatility models can also be represented by fractional Brownian motion (see e.g.~\cite{bayer0, gander}), but there are reasons for preferring SVEs; see e.g.~\cite{richard} for instance.

In many cases of practical interest the SVE can be written as a type of stochastic differential equation (SDE), whose solution exists and is unique (under assumptions), but which cannot be recovered analytically and often not simulated; see e.g.~\cite{richard} and the references therein. As a result there is great interest in time-discretizing the SDE and working with a biased, but convergent approximation scheme. There have been several recent works in this direction including \cite{bayer,richard} of which we follow the latter and consider the Euler-Maruyama (EM) method. Adopting EM and given a prior on the unknown parameters, gives rise to a posterior density on the parameters and the time-discretized path of the SVE, which is non-Markovian and has a bias due to time-discretization.

Performing fully Bayesian inference in the scenario that the partially observed process is Markovian, e.g.~conventional Brownian and L\'evy driven SDEs has already received a great deal of attention in the literature; see \cite{andrieu,delta,jasra,jasra_levy, hamza, pierre} for a non-exhaustive list. This typically depends on designing an appropriate Markov chain Monte Carlo (MCMC) scheme to sample from the posterior density of the parameters and latent process. In later work \cite{jasra} amongst others leveraged the multilevel Monte Carlo (MLMC) method \cite{giles,giles1,hein}, which works with a hierarchy of time discretizations and under appropriate mathematical assumptions can reduce the computational effort to achieve a pre-specified mean square error (MSE) versus using a single discretization and MCMC; see \cite{ml_rev} for a review in the context of data problems. \cite{jasra} combine MLMC and MCMC in the case of regular diffusion processes.

The work that combines MLMC with non-Markovian time discretized dynamics with MCMC and MLMC is far rarer and to the best of our knowledge we are only aware of \cite{jasra_fbm} in the context of fractional Brownian motion. The often, extra overhead in computation, versus Markovian processes, makes the application of cost reduction methods such as MLMC far more important. In this article we develop a MCMC method that combines with MLMC to facilitate fully Bayesian inference for a particular class of partially observed SVE models.
In particular, and under assumptions, we prove that the cost to achieve a MSE associated to posterior expectations of $\mathcal{O}(\epsilon^2)$, $\epsilon>0$, is $\mathcal{O}(\epsilon^{-\tfrac{20}{9}})$. If one uses a single level MCMC method then the cost is 
$\mathcal{O}(\epsilon^{-\tfrac{38}{9}})$ to achieve the same MSE. This is for a particular model for which we also give numerical simulations on both synthetic and real data.

This article is structured as follows. In Section \ref{sec:model} we describe the class of model that we consider as well as the associated time discretization. In Section \ref{sec:alg_anal} we give our algorithms for statistical inference as well as a mathematical analysis thereof. Section \ref{sec:num} provides our numerical results. Our mathematical proofs are housed in the appendix.

\section{Model}
\label{sec:model}

Let $(\Omega,\mathcal{F},\mathbb{P})$ be a probability space and define the following process on it:
\begin{eqnarray}
V_t & = & V_0 + \int_0^t K(t-s)\left\{\left(\kappa-\lambda V_s\right)ds + \nu \sqrt{V}_sdW_s\right\} \label{eq:vol}
\end{eqnarray}
where $\{W_t\}_{t\geq 0}$ is a one-dimensional Brownian motions, $K(t)=Ct^{H}$, {where $V_0>0$ is the starting point, $\lambda>0$ is the speed of the mean-reversion to the long-term level $\kappa>0$. The parameter $\nu>0$ is the volatility coefficient. $H \in [0,1/2)$ is the Hurst parameter that is used to measure the roughness of the path.} This model is very similar to the latent (volatility) process that is used in the rough Heston model e.g.~\cite{euch}, {which is a generalization of the renowned Heston model \cite{heston}}. It is explicitly assumed that there is a unique solution and we do not mention this point further.

We consider discrete time observations $Y_1,\dots,Y_T$ with some parametric density and conditionally independent given some characteristics of $\{V_t\}_{t\geq 0}$. We give two examples below. {We note that the parameter $H$ will be fixed for now, as it plays a crucial role in determining the model's behavior. In Section \ref{subsec:synth_data}, the objective of the simulations will be to numerically compare the rates between the single and multilevel approaches in the following two examples, namely, the state-space and the stochastic volatility models. Therefore, to ensure a proper convergence of the Markov chain, we fix the value of $H$. However, in Section \ref{subsec:real_data}, where real data is considered, we estimate the value of $H$ in addition to the rest of the parameters.}

\subsection{State-Space Model}\label{sec:ssm}

In this case we simply take $Y_t|V_t$ independent of all other random variables, with some parametric density $g(\cdot|v_t)$ and set $\theta:=(V_0,\kappa,\lambda,\nu)\in(\mathbb{R}^+)^4:=\Theta$. The objective is, given a proper prior $\pi$ on $\Theta$ to draw Bayesian inference in $\theta$ and (at least) the skeleton of volatilities $V_{1},\dots,V_{T}$:
$$
\pi(\theta,v_{1:T}|y_{1:T}) \propto \left(\prod_{t=1}^T g(y_t|v_t) p(v_t|v_{t-1},\theta)\right)\pi(\theta)
$$
with $\pi(\theta)$ the prior. 

\subsection{Stochastic Volatility Model}

Let $\{Y_t\}$ be a continuous-time process 
\begin{equation}
Y_t = Y_0 + \int_0^t rd s + \int_{0}^t\sqrt{V_s}d\left\{\rho W_t + \sqrt{1-\rho^2}B_t\right\}, \label{eq:log_price}
\end{equation}
where $\{B_t\}_{t\geq 0}$ is a one-dimensional standard Brownian motion independent of $\{W_t\}_{t\geq 0}$. Throughout we will take $Y_0=y_0$ as given.
In continuous time and conditional upon 
$$
\int_{0}^{1}\sqrt{V_s} dW_s,\dots, \int_{T-1}^{T}\sqrt{V_s} dW_s,\int_{0}^1 V_s ds,\dots,\int_{T-1}^T V_s ds,\theta
$$
we have that the conditional likelihood of $T$ discretely observed data at unit times is
$$
p\left(y_{1:T}\Big|\int_{0}^{1}\sqrt{V_s} dW_s,\dots, \int_{T-1}^{T}\sqrt{V_s} dW_s,\int_{0}^1 V_s ds,\dots,\int_{T-1}^T V_s ds,\theta\right) = 
$$
\begin{equation}\label{eq:exact_like_rh}
\prod_{t=1}^T \psi\left(y_t;y_{t-1} + r + \rho\int_{t-1}^t \sqrt{V_s} dW_s, (1-\rho^2)\int_{t-1}^t V_s ds\right)
\end{equation}
where $\psi(y;\mu,\sigma^2)$ is the one-dimensional Gaussian density of mean $\mu$ and variance $\sigma^2$. {In this context, $Y_t$ represents a random variable and $\{Y_t\}$ a continuous-time observation process, while $y_{1:T}$ refers to the actual observations received at unit times.} Setting
$\theta:=(V_0,\rho,\kappa,\lambda,\nu,r)\in\mathbb{R}^+\times[-1,1]\times(\mathbb{R}^+)^3\times\mathbb{R}:=\Theta$, one again can consider Bayesian inference with a proper prior on $\Theta$. {In Section \ref{sec:num}, we employ Gaussian prior distributions.}

\subsection{Discretization}\label{sec:disc}

We give the Euler-Maruyama discretization method. In the case of state-space models, the resulting inference is simple.
For the stochastic volatility (SV) model, more work is needed and detailed below.

\subsubsection{Euler Discretization}\label{sec:euler}

{Let $\Delta_l=2^{-l}$ be the stepsize.} 
We follow \cite{richard} and consider the Euler discretization of 
\eqref{eq:vol} as
\begin{eqnarray}
V_{(k+1)\Delta_l} & = & V_0 + \sum_{j=0}^k \Big\{K([k+1-j]\Delta_L)\left(\kappa-\lambda V_{j\Delta_l}\right)\Delta_l
+ K([k+1-j]\Delta_l)\nu\times\nonumber \\ & & \sqrt{|V_{j\Delta_l}|}\left(W_{(j+1)\Delta_l}-
W_{j\Delta_l}\right)\Big\}.\label{eq:vol_disc}
\end{eqnarray}
A more sophisticated discretization of the Volterra SDE in \eqref{eq:vol} is can be found in \cite{volt_disc}. We remark that the cost of implementing this discretization is at least $\mathcal{O}(\Delta_l^{-3})$, if exact Cholesky decompositions are used. It does not seem that (e.g.) circulant methods (see \cite{gauss}) can be used. \cite{volt_disc1} provide an approximation of the discretization that can reduce the computational cost, but as the cost of the application of Euler methods (which are `exact') is already $\mathcal{O}(\Delta_l^{-2})$ we only consider this case in this article.

\subsubsection{Stochastic Volatility Model}

Denote the increments of the Brownian motion $\{W_t\}_{t\geq 0}$ as $\overline{W}_{\Delta_l},\dots,\overline{W}_{T}$ and we remark that, independently for each $k\in\{1,\dots,T\Delta_l^{-1}\}$
$$
\overline{W}_{k\Delta_l} \sim \mathcal{N}(0,\Delta_l),
$$
where $\mathcal{N}(0,\Delta_l)$ denotes the one-dimensional Gaussian distribution of mean 0 and variance $\Delta_l$.

In practice one can replace the integrals in \eqref{eq:exact_like_rh} with a discretized process
$$
p^l\left(y_{1:T}\Big|\overline{w}_{\Delta_l:T},\theta\right) = 
$$
\begin{equation}\label{eq:disc_like}
\prod_{t=1}^T \psi\left(y_t;y_{t-1} + \rho \sum_{k=0}^{\Delta_l^{-1}} \sqrt{|v_{t-1+k\Delta_l}|} 
\overline{w}_{t-1+k\Delta_l}, (1-\rho^2)\Delta_l\sum_{k=0}^{\Delta_l^{-1}}|v_{t-1+k\Delta_l}|\right).
\end{equation}
for any discretized path $v_{\Delta_l:T}$. Thus, we have the posterior
\begin{align}\label{eq:post_euler}
\pi^l(\theta,\overline{w}_{\Delta_L:T}|y_{1:T}) &\propto 
p^l\left(y_{1:T}\Big|\overline{w}_{\Delta_l:T},\theta\right)
p(\overline{w}_{\Delta_l:T})\pi(\theta)
\end{align}
with $\pi(\theta)$ also denoting the prior, although it is different than in the state-space model case.

\section{Algorithms and Analysis}\label{sec:alg_anal}

In the following section all of our presentation and analysis focusses upon the discretization detailed in Section \ref{sec:euler} and in particular the posterior \eqref{eq:post_euler} associated to the likelihood \eqref{eq:disc_like}. However, extension of the algorithms to the case of the state-space model, is more-or-less standard and is omitted for brevity.

\subsection{Particle Filter and Delta-Particle Filter}

We begin by presenting the particle filter associated to the model in \eqref{eq:post_euler} except that $\theta$ is fixed. As a method for filtering, one would not expect it to work very well, but is an essential component in the subsequent algorithms. To simplify the notation we shall use the following shorthand for $(t,l)\in\{1,\dots,T\}\times\mathbb{N}_0$ given:
$$
\kappa_{t,l}(\overline{w}_{\Delta_l:t}) := \psi\left(y_t;y_{t-1} + \rho \sum_{k=0}^{\Delta_l^{-1}} \sqrt{|v_{t-1+k\Delta_l}|} 
\overline{w}_{t-1+k\Delta_l}, (1-\rho^2)\Delta_l\sum_{k=0}^{\Delta_l^{-1}}|v_{t-1+k\Delta_l}|\right).
$$
The particle filer is given in Algorithm \ref{alg:pf} as it is required in the sequel. 

\begin{algorithm}[h]
\caption{Particle Filter}
\label{alg:pf}
\begin{enumerate}
\item{Input: data $y_{1:T}$, level $l\in\mathbb{N}_0$, particle number $N\in\mathbb{N}$ and parameter $\theta\in\Theta$.}
\item{Initialize: For $i\in\{1,\dots,N\}$, independently generate $\overline{W}_{\Delta_l:1}^i$ from $\mathcal{N}(0,\Delta_l)$. Set $t=1$, {$\hat{p}^N(y_{1:0})=1$ for convention} and go to step 3.}
\item{Iterate: For $i\in\{1,\dots,N\}$ compute
$$
u_t^i = \frac{\kappa_{t,l}(\overline{w}_{\Delta_l:t}^i)}{\sum_{j=1}^N\kappa_{t,l}(\overline{w}_{\Delta_l:t}^j)}.
$$
Set $\hat{p}^N(y_{1:t})={\hat{p}^N(y_{1:t-1})}\tfrac{1}{N}\sum_{i=1}^N\kappa_{t,l}(\overline{w}_{\Delta_l:t}^i)$.
Then sample $\overline{w}_{\Delta_l:t}^{1:N}$ with replacement from  $\overline{w}_{\Delta_l:t}^{1:N}$ using probabilities
$u_t^{1:N}$. For $i\in\{1,\dots,N\}$, independently generate $\overline{W}_{t+\Delta_l:t+1}^i$ from $\mathcal{N}(0,\Delta_l)$.
Set $t=t+1$ and if $t=T$ go to step 4, otherwise restart step 3.}
\item{Grand Selection: For $i\in\{1,\dots,N\}$ compute
$
u_T^i = \frac{\kappa_{T,l}(\overline{w}_{\Delta_l:T}^i)}{\sum_{j=1}^N\kappa_{T,l}(\overline{w}_{\Delta_l:T}^j)}.
$
Set $\hat{p}^N(y_{1:T})=\hat{p}^N(y_{1:T-1})\tfrac{1}{N}\sum_{i=1}^N\kappa_{t,l}(\overline{w}_{\Delta_l:T}^i)$.
Sample one $\overline{w}_{\Delta_l:T}$ from $\overline{w}_{\Delta_l:T}^{1:N}$ using $u_T^{1:N}$ and go to step 5.}
\item{Output: trajectory $\overline{w}_{\Delta_l:T}$ and normalizing constant estimate $\hat{p}^N(y_{1:T})$.}
\end{enumerate}
\end{algorithm}

The delta-particle filter that was developed in \cite{jasra} (and renamed in \cite{delta, hamza}) is, on its own, not a particularly 
useful algorithm. However, in our context, it is vital for our multilevel MCMC algorithm to be presented below. The delta-particle filter is given in Algorithm \ref{alg:dpf}. {In the delta-particle filter, each particle is associated with two consecutive discretization levels. These levels are correlated in a specific way: the coarse-grained Brownian motion increment is the sum of adjacent fine-grained increments within the same time step. Unlike the particle filter, which updates the particle weights and the normalizing constant by calculating the likelihood between the observation and the particle state, the delta-particle filter constructs an approximate coupling of the likelihood density of the hidden state at two consecutive discretization levels.}


\begin{algorithm}[h]
\caption{Delta Particle Filter}
\label{alg:dpf}
\begin{enumerate}
\item{Input: data $y_{1:T}$, level $l\in\mathbb{N}$, particle number $N\in\mathbb{N}$ and parameter $\theta\in\Theta$.}
\item{Initialize: For $i\in\{1,\dots,N\}$, independently generate $\overline{W}_{\Delta_l:1}^{i,l}$ from $\mathcal{N}(0,\Delta_l)$. For $(i,k)\in\{1,\dots,N\}\times\{1,\dots,\Delta_{l-1}^{-1}\}$ set
$
\overline{W}_{k\Delta_{l-1}}^{i,l-1} = \overline{W}_{2k\Delta_{l}}^{i,l} + \overline{W}_{(2k-1)\Delta_{l}}^{i,l}.
$
Set $t=1$, $\hat{\tilde{p}}^N(y_{1:0})=1$ for convention and go to step 3.}
\item{Iterate: For $i\in\{1,\dots,N\}$ compute
$$
\tilde{u}_t^i = \frac{\max\{\kappa_{t,l}(\overline{w}_{\Delta_l:t}^{i,l}),\kappa_{t,l-1}(\overline{w}_{\Delta_{l-1}:t}^{i,l-1})\}}{\sum_{j=1}^N\max\{\kappa_{t,l}(\overline{w}_{\Delta_l:t}^{j,l}),\kappa_{t,l-1}(\overline{w}_{\Delta_{l-1}:t}^{j,l-1})\}}.
$$
Set $\hat{p}^N(y_{1:t})={\hat{p}^N(y_{1:t-1})}\tfrac{1}{N}\sum_{i=1}^N\max\{\kappa_{t,l}(\overline{w}_{\Delta_l:t}^{i,l}),\kappa_{t,l-1}(\overline{w}_{\Delta_{l-1}:t}^{i,l-1})\}$.
Then sample $(\overline{w}_{\Delta_l:t}^{1:N,l},\overline{w}_{\Delta_{l-1}:t}^{1:N,l-1})$ with replacement from  
$(\overline{w}_{\Delta_l:t}^{1:N,l},\overline{w}_{\Delta_{l-1}:t}^{1:N,l-1})$
using probabilities
$\tilde{u}_t^{1:N}$. For $i\in\{1,\dots,N\}$, independently generate $\overline{W}_{t+\Delta_l:t+1}^i$ from $\mathcal{N}(0,\Delta_l)$. For $(i,k)\in\{1,\dots,N\}\times\{1,\dots,\Delta_{l-1}^{-1}\}$ set
$
\overline{W}_{t+k\Delta_{l-1}}^{i,l-1} = \overline{W}_{t+2k\Delta_{l}}^{i,l} + \overline{W}_{t+(2k-1)\Delta_{l}}^{i,l}.
$
Set $t=t+1$ and if $t=T$ go to step 4, otherwise restart step 3.}
\item{Grand Selection: For $i\in\{1,\dots,N\}$ compute
$
\tilde{u}_T^i = \frac{\max\{\kappa_{T,l}(\overline{w}_{\Delta_l:T}^{i,l}),\kappa_{T,l-1}(\overline{w}_{\Delta_{l-1}:T}^{i,l-1})\}}{\sum_{j=1}^N\max\{\kappa_{T,l}(\overline{w}_{\Delta_l:T}^{j,l}),\kappa_{T,l-1}(\overline{w}_{\Delta_{l-1}:T}^{j,l-1})\}}.
$
Set $\hat{p}^N(y_{1:T})=\hat{p}^N(y_{1:T-1})\tfrac{1}{N}\sum_{i=1}^N
\max\{\kappa_{T,l}(\overline{w}_{\Delta_l:T}^{i,l}),\kappa_{T,l-1}(\overline{w}_{\Delta_{l-1}:T}^{i,l-1})\}$.
Sample one $(\overline{w}_{\Delta_l:T}^l, \overline{w}_{\Delta_{l-1}:T}^{l-1})$ from $(\overline{w}_{\Delta_l:T}^{1:N,l},\overline{w}_{\Delta_{l-1}:T}^{1:N,l-1})$ using $\tilde{u}_T^{1:N}$ and go to step 5.}
\item{Output: trajectories $(\overline{w}_{\Delta_l:T}^l,\overline{w}_{\Delta_{l-1}:T}^{l-1})$ and normalizing constant estimate $\hat{\tilde{p}}^N(y_{1:T})$.}
\end{enumerate}
\end{algorithm}

\subsection{Particle MCMC}

We now describe the first method that is needed for our approach. This is the basic particle MCMC (PMCMC) developed in \cite{andrieu} simply written in our notation and is presented in Algorithm \ref{alg:pmcmc}. The output from Algorithm \ref{alg:pmcmc} can be used to estimate expectations with respect to (w.r.t.)~$\pi^l$ as in \eqref{eq:post_euler}.

\begin{algorithm}[h]
\caption{Particle MCMC}
\label{alg:pmcmc}
\begin{enumerate}
\item{Input: data $y_{1:T}$, level $l\in\mathbb{N}_0$, particle number $N\in\mathbb{N}$, iteration number $M\in\mathbb{N}$ and proposal $q_l$.}
\item{Initialize: Sample $\theta_0^{l}$ from the prior and then run Algorithm \ref{alg:pf} with parameter $\theta_0^{l}$ to give $\overline{w}_{0,\Delta_l:T}^l$, denoting the normalizing constant estimate $\hat{p}_{\theta_0^{l}}^N(y_{1:T})$. Set $k=1$.}
\item{Iterate:
Sample $\theta'|\theta_{k-1}^{l}$ from the proposal $q_l(\cdot|\theta_{k-1}^{l})$ and then run Algorithm \ref{alg:pf} with parameter $\theta'$, denoting the normalizing constant estimate $\hat{p}_{\theta'}^N(y_{1:T})$ and the proposed path {$\overline{w}_{\theta',\Delta_l:T}^l$}.
Set $\theta_k^{l}=\theta'$, $\overline{w}_{k,\Delta_l:T}^l={\overline{w}_{\theta',\Delta_l:T}^l}$ and $\hat{p}_{\theta_k^{l}}^N(y_{1:T})=\hat{p}_{\theta'}^N(y_{1:T})$ with probability
$$
\min\left\{1,\frac{\hat{p}_{\theta'}^N(y_{1:T})\pi(\theta')q_l(\theta_k^{l}|\theta')}{\hat{p}_{\theta_k^{l}}^N(y_{1:T})\pi(\theta_k^{l})q_l(\theta'|\theta_k^{l})}\right\}
$$
otherwise set $\theta_k^{l}=\theta_{k-1}^l$, $\overline{w}_{k,\Delta_l:T}^l=\overline{w}_{k-1,\Delta_l:T}^l$ and $\hat{p}_{\theta_k^{l}}^N(y_{1:T})=\hat{p}_{\theta_{k-1}^l}^N(y_{1:T})$. Set $k=k+1$ and if $k=M+1$ go to step 4, otherwise go to the start to step 3.}
\item{Output: $(\theta_{0:M}^{l},\overline{w}_{0:M,\Delta_l:T}^l)$.}
\end{enumerate}
\end{algorithm}

\subsection{Multilevel MCMC and Estimator}

We are now in a position to give our final algorithm and estimator. We begin with a type of coupled MCMC which is presented in Algorithm \ref{alg:cmcmc}. It can be shown (e.g.~\cite{jasra}) that the samples from Algorithm \ref{alg:cmcmc} will correspond to that of an ergodic Markov chain of invariant measure
$$
\tilde{\pi}^{l,l-1}(d(\theta,\overline{w}_{\Delta_l:T}^l(l),\overline{w}_{\Delta_{l-1}:T}^{l-1}(l))) \propto
$$
$$
\left(\prod_{t=1}^T \max\{\kappa_{t,l}(\overline{w}_{\Delta_l:t}^l),\kappa_{t,l-1}(\overline{w}_{\Delta_{l-1}:t}^{l-1})\}\right)
\pi(\theta)d\theta\tilde{Q}^{l,l-1}(d(\overline{w}_{\Delta_l:T}^l(l),\overline{w}_{\Delta_{l-1}:T}^{l-1}(l)))
$$
where $\tilde{Q}^{l,l-1}(d(\overline{w}_{\Delta_l:T}^l(l),\overline{w}_{\Delta_{l-1}:T}^{l-1}(l)))$ is the (probability law) of the  synchronous coupling of the Brownian increments that is used in Algorithm \ref{alg:dpf} (step 2 and 3) and $d\theta$ is some $\sigma-$finite dominating measure on $(\Theta,\mathcal{B}(\Theta))$ such as Lebesgue. 
Now, define for $l\in\mathbb{N}$:
\begin{eqnarray*}
H_{1,l}(\overline{w}_{\Delta_l:T}^l,\overline{w}_{\Delta_{l-1}:T}^{l-1}) & := & \prod_{t=1}^T \frac{\kappa_{t,l}(\overline{w}_{\Delta_l:t}^l)}{\max\{\kappa_{t,l}(\overline{w}_{\Delta_l:t}^l),\kappa_{t,l-1}(\overline{w}_{\Delta_{l-1}:t}^{l-1})\}} \\
H_{2,l}(\overline{w}_{\Delta_l:T}^l,\overline{w}_{\Delta_{l-1}:T}^{l-1}) & := & \prod_{t=1}^T \frac{\kappa_{t,l-1}(\overline{w}_{\Delta_{l-1}:t}^{l-1})}{\max\{\kappa_{t,l}(\overline{w}_{\Delta_l:t}^l),\kappa_{t,l-1}(\overline{w}_{\Delta_{l-1}:t}^{l-1})\}}.
\end{eqnarray*}
Then for $\varphi:\Theta\times\mathbb{R}^{T}\rightarrow\mathbb{R}$, we have the identity:
\begin{equation}\label{eq:main_id}
\mathbb{E}_{\pi^l}[\varphi(\theta,V_{1:T})] - \mathbb{E}_{\pi^{l-1}}[\varphi(\theta,V_{1:T})] = 
\end{equation}
$$
\frac{\mathbb{E}_{\tilde{\pi}^{l,l-1}}[\varphi(\theta,V_{1:T}^l)H_{1,l}(\overline{W}_{\Delta_l:T}^l,\overline{W}_{\Delta_{l-1}:T}^{l-1})]}{\mathbb{E}_{\tilde{\pi}^{l,l-1}}[H_{1,l}(\overline{W}_{\Delta_l:T}^l,\overline{W}_{\Delta_{l-1}:T}^{l-1})]} - 
\frac{\mathbb{E}_{\tilde{\pi}^{l,l-1}}[\varphi(\theta,V_{1:T}^{l-1})H_{2,l}(\overline{W}_{\Delta_l:T}^l,\overline{W}_{\Delta_{l-1}:T}^{l-1})]}{\mathbb{E}_{\tilde{\pi}^{l,l-1}}[H_{2,l}(\overline{W}_{\Delta_l:T}^l,\overline{W}_{\Delta_{l-1}:T}^{l-1})]}
$$
where the subscript in the expectation operator denotes the probability law which the expectation is taken w.r.t.,~and 
$V_{1:T}^l,V_{1:T}^{l-1}$ in the functional $\varphi$ are understood as the discretized $V_{1:T}$ from \eqref{eq:vol_disc}  (at levels $l$ and $l-1$) which are functionals of the Brownian increments. We remark that we need not consider dependence only upon $V_{1:T}$ in the functional of interest, but it is convenient to do so.

Given the identity \eqref{eq:main_id} then one can use Algorithm \ref{alg:cmcmc} to approximate $\mathbb{E}_{\pi^l}[\varphi(\theta,V_{1:T})] - \mathbb{E}_{\pi^{l-1}}[\varphi(\theta,V_{1:T})]$ with
$$
\left\{\widehat{\mathbb{E}_{\pi^l}[\varphi(\theta,V_{1:T})] - \mathbb{E}_{\pi^{l-1}}[\varphi(\theta,V_{1:T})]}\right\}^M = 
$$
$$
\tfrac{\tfrac{1}{M+1}\sum_{k=0}^M\varphi(\theta_k^l,v_{k,1:T}^l(l))H_{1,l}(\theta_k^l,\overline{w}_{k,\Delta_l:T}^l(l),
\overline{w}_{k,\Delta_{l-1}:T}^{l-1}(l)
)}{\tfrac{1}{M+1}\sum_{k=0}^MH_{1,l}(\theta_k^l,\overline{w}_{k,\Delta_l:T}^l(l),
\overline{w}_{k,\Delta_{l-1}:T}^{l-1}(l)
)} -
\tfrac{\tfrac{1}{M+1}\sum_{k=0}^M\varphi(\theta_k^l,v_{k,1:T}^{l-1}(l))H_{2,l}(\theta_k^l,\overline{w}_{k,\Delta_l:T}^l(l),
\overline{w}_{k,\Delta_{l-1}:T}^{l-1}(l)
)}{\tfrac{1}{M+1}\sum_{k=0}^MH_{2,l}(\theta_k^l,\overline{w}_{k,\Delta_l:T}^l(l),
\overline{w}_{k,\Delta_{l-1}:T}^{l-1}(l)
)} 
$$
where $v_{k,1:T}^l(l)$ has been obtained by running the recursion \eqref{eq:vol_disc} with increments $\overline{w}_{k,\Delta_l:T}^l(l)$ at level $l$ and a similar notation for $v_{k,1:T}^{l-1}(l)$ has been obtained by running the recursion \eqref{eq:vol_disc} with increments $\overline{w}_{k,\Delta_{l-1}:T}^{l-1}(l)$ at level $l-1$.

Given the above exposition, we shall now describe our final method. We seek to approximate the multilevel identity
$$
\mathbb{E}_{\pi^L}[\varphi(\theta,V_{1:T})] = \mathbb{E}_{\pi^0}[\varphi(\theta,V_{1:T})] + \sum_{l=1}^L\{
\mathbb{E}_{\pi^l}[\varphi(\theta,V_{1:T})] - \mathbb{E}_{\pi^{l-1}}[\varphi(\theta,V_{1:T})]\}.
$$
where $L\in\mathbb{N}$ has to be specified ultimately, but we shall take it as given for now. Given $(M_0,\dots,M_L)\in\mathbb{N}^{L+1}$ one can then estimate $\mathbb{E}_{\pi^L}[\varphi(\theta,V_{1:T})]$ using the following approach.
\begin{enumerate}
\item{Run Algorithm \ref{alg:pmcmc} with $M_0$ iterations and at level $0$ returning the estimate:
$$
\left\{\widehat{\mathbb{E}_{\pi^0}[\varphi(\theta,V_{1:T})]}\right\}^{M_0} = \frac{1}{M_0+1}\sum_{k=0}^{M_0}\varphi(\theta_k^0,v_{k,1:T}^{0}).
$$
}
\item{For $l\in\{1,\dots,L\}$ independently and independently of step 1, run  Algorithm \ref{alg:cmcmc} with $M_l$ iterations and at level $l$ returning the estimate $\left\{\widehat{\mathbb{E}_{\pi^l}[\varphi(\theta,V_{1:T})] - \mathbb{E}_{\pi^{l-1}}[\varphi(\theta,V_{1:T})]}\right\}^{M_l}$.}
\end{enumerate}
Then one has the multilevel particle MCMC (MLPMCMC) approximation of $\mathbb{E}_{\pi^L}[\varphi(\theta,V_{1:T})]$:
\begin{equation}\label{eq:main_eq}
\left\{\widehat{\mathbb{E}_{\pi^L}[\varphi(\theta,V_{1:T})]}\right\}^{M_{0:L}} :=
\left\{\widehat{\mathbb{E}_{\pi^0}[\varphi(\theta,V_{1:T})]}\right\}^{M_0}+ \sum_{l=1}^L \left\{
\left\{\widehat{\mathbb{E}_{\pi^l}[\varphi(\theta,V_{1:T})] - \mathbb{E}_{\pi^{l-1}}[\varphi(\theta,V_{1:T})]}\right\}^{M_l}
\right\}.
\end{equation}
The main issue now is to choose $L$ and $M_{0:L}$.

\begin{algorithm}[h]
\caption{Coupled MCMC}
\label{alg:cmcmc}
\begin{enumerate}
\item{Input: data $y_{1:T}$, level $l\in\mathbb{N}$, particle number $N\in\mathbb{N}$, iteration number $M\in\mathbb{N}$ and proposal $q_l$.}
\item{Initialize: Sample $\theta_0^{l}$ from the prior and then run Algorithm \ref{alg:dpf} with parameter $\theta_0^{l}$ to give $(\overline{w}_{0,\Delta_l:T}^l(l),\overline{w}_{0,\Delta_{l-1}:T}^{l-1}(l))$ , denoting the normalizing constant estimate $\hat{\tilde{p}}_{\theta_0^{l}}^N(y_{1:T})$. Set $k=1$.}
\item{Iterate:
Sample $\theta'|\theta_{k-1}^{l}$ from the proposal $q_l(\cdot|\theta_{k-1}^{l})$ and then run Algorithm \ref{alg:dpf} with parameter $\theta'$, denoting the normalizing constant estimate $\hat{\tilde{p}}_{\theta'}^N(y_{1:T})$ and proposed paths $({\overline{w}_{\theta',\Delta_l:T}^l},\overline{w}_{\theta',\Delta_{l-1}:T}^{l-1})$.
Set $\theta_k^{l}=\theta'$, $(\overline{w}_{k,\Delta_l:T}^l(l),\overline{w}_{k,\Delta_{l-1}:T}^{l-1}(l))=({\overline{w}_{\theta',\Delta_l:T}^l},\overline{w}_{\theta',\Delta_{l-1}:T}^{l-1})$ and $\hat{\tilde{p}}_{\theta_k^{l}}^N(y_{1:T})=\hat{\tilde{p}}_{\theta'}^N(y_{1:T})$ with probability
$$
\min\left\{1,\frac{\hat{\tilde{p}}_{\theta'}^N(y_{1:T})\pi(\theta')q_l(\theta_k^{l}|\theta')}{\hat{\tilde{p}}_{\theta_k^{l}}^N(y_{1:T})\pi(\theta_k^{l})q_l(\theta'|\theta_k^{l})}\right\}
$$
otherwise set $\theta_k^{l}=\theta_{k-1}^l$, $(\overline{w}_{k,\Delta_l:T}^l(l),\overline{w}_{k,\Delta_{l-1}:T}^{l-1}(l))=(\overline{w}_{k-1,\Delta_l:T}^l(l),\overline{w}_{k-1,\Delta_{l-1}:T}^{l-1}(l))$ and $\hat{\tilde{p}}_{\theta_k^{l}}^N(y_{1:T})=\hat{\tilde{p}}_{\theta_{k-1}^l}^N(y_{1:T})$. Set $k=k+1$ and if $k=M+1$ go to step 4, otherwise go to the start to step 3.}
\item{Output: $(\theta_{0:M}^{l},\overline{w}_{0:M,\Delta_l:T}^l(l),\overline{w}_{0:M,\Delta_{l-1}:T}^{l-1}(l))$.}
\end{enumerate}
\end{algorithm}

\subsection{Mathematical Analysis}

We now give a mathematical analysis associated to our estimator \eqref{eq:main_eq} which will assist us in choosing $L$ and $M_{0:L}$. {In the discussion section following Proposition \ref{prop:main_res1}, we present the approach that we employ to choose $L$ and $M_{0:L}$.} The analysis for a state-space model (i.e.~as in Section \ref{sec:ssm}) using an Euler discretization is identical to that in \cite{jasra}; as a result we consider a simplified version of the SV model. Although this model is not exactly the one that we will use in our numerical results, the guidelines that we derive from our analysis should illuminate the general case.

The model considered is taken as for $t\in\{1,\dots,T\}$, $y_t\in\mathsf{Y}$
$$
p(y_t|\{v_s\}_{s\in[0,t]},y_{1:t-1},y_{t+1:T},\theta) = \overline{g}\left(y_t;y_{t-1},\int_{t-1}^t \overline{f}(v_s)ds,\int_{t-1}^t \overline{h}(v_s)dW_s\right)
$$
where $\overline{f},\overline{h}:\mathbb{R}\rightarrow\mathbb{R}$, $\{V_t\}_{t\in[0,T]}$ follows a generic Volterra SDE (such as the process \eqref{eq:vol}) and $\overline{g}$ is a probability density in $y_t$ with $y_0$ given. We will use an Euler discretization such as in \eqref{eq:vol_disc} and then consider the discretized likelihood model
$$
\overline{p}^l\left(y_{1:T}|\overline{w}_{\Delta_l:T},\theta\right) = 
\prod_{t=1}^T  \overline{g}\left(y_t;y_{t-1},\Delta_l\sum_{k=0}^{\Delta_l^{-1}-1} \overline{f}(v_{t-1+k\Delta_l}),\sum_{k=0}^{\Delta_l^{-1}-1} \overline{h}(v_{t-1+k\Delta_l})\overline{w}_{t-1+(k+1)\Delta_l}\right).
$$
We shall take $V_0$ fixed and known.

The Markov kernel associated to the iterate step (Step 3.) of Algorithm \ref{alg:cmcmc} is denoted as $K_l$ (appropriately modified to the model described above) and that these are defined on some measurable space $(\mathsf{W}_l,\mathcal{W}_l)$ which can be made precise (see e.g.~\cite{andrieu,jasra}).
We assume that our Markov chain(s) are started in stationarity; this point is considered (and explained) in \cite{jasra}. We make the following assumptions to facilitate our analysis.
Below $\mathcal{B}_b(\mathsf{X})$ (resp.~$\textrm{Lip}(\mathsf{X})$) are the collection of bounded and measurable (resp.~Lipschitz where we take the $L_1-$distance as the norm) real-valued functions on a measurable space $(\mathsf{X},\mathscr{X})$.

\newtheorem{assumption}{Assumption}
\begin{assumption}
\label{ass:1}
\begin{enumerate}
\item{$(\overline{f},\overline{g})\in(\textrm{Lip}(\mathbb{R})\cap\mathcal{B}_b(\mathbb{R}))^2$.}
\item{For any $(y,y')\in\mathsf{Y}^2$ there exists a constant $C<\infty$ such that for any $(x_{1:2},x_{1:2}')\in\mathbb{R}^4$:
$$
|\overline{g}(y;y',x_1,x_2)-\overline{g}(y;y',x_1',x_2')| \leq C\left(|x_1-x_1'| +|x_2-x_2'|\right).
$$
}
\item{For any $(y,y')\in\mathsf{Y}^2$ there exists constants $0<\underline{C}<\overline{C}<\infty$ such that for any $x_{1:2}\in\mathbb{R}^2$:
$$
\underline{C}\leq \overline{g}(y;y',x_1,x_2)\leq \overline{C}.
$$}
\end{enumerate}
\end{assumption}
To avoid confusion, in the below assumption $\{V_t\}_{t\in[0,T]}$ is a Volterra SDE process and $\{V_t^l\}_{t\in\{0,\Delta_l,\dots,T\}}$ is the Euler discretized process. We use $\mathbb{E}$ to denote the expectation operator associated to these processes.
\begin{assumption}
\label{ass:2}
For $r\in[1,2]$ there exists a $C<\infty$ such that for any $(l,t)\in\mathbb{N}_0\times\{\Delta_l,2\Delta_l,\dots,T\}$
$$
\sup_{\theta\in\Theta}\mathbb{E}[|V_t^l-V_t|^r]^{1/r} \leq C\Delta_l^{\beta}.
$$
\end{assumption}
Below we write the collection of probability measures on $(\mathsf{W}_l,\mathcal{W}_l)$ as $\mathscr{P}(\mathsf{W}_l)$.
\begin{assumption}
\label{ass:3}
There exist a $\xi\in(0,1)$ and a $\nu\in\mathscr{P}(W_l)$ such that, for each $w\in\mathsf{W}_l$, $l\in\mathbb{N}$,
$$
K_l(w,dw') \geq \xi\nu(dw').
$$
In addition $K_l$ is a $\eta_l-$reversible kernel where $\eta_l$ is the invariant measure of the kernel.
\end{assumption}

The above assumptions are fairly strong. Assumptions \ref{ass:2}-\ref{ass:3} are discussed (and verified) for state-space models driven by diffusions processes. Assumption \ref{ass:2} is investigated extensively in \cite{richard}.
Assumption \ref{ass:1} essentially states that the likelihood function is fairly regular and allows one to consider the convergence properties of the driving process.

We have the following result, whose proof is given in Appendix \ref{app:1}.  $\mathbb{E}$ is used to denote expectation w.r.t.~the probability associated to the random process used in our numerical approach.
\begin{prop}\label{prop:main_res}
Assume Assumptions \ref{ass:1}-\ref{ass:3}. Then for any $\varphi\in\textrm{\emph{Lip}}(\Theta\times\mathbb{R}^T)\cap\mathcal{B}_b(\Theta\times\mathbb{R}^T)$ there exists a $C<\infty$ such that for any $(l,M_l)\in\mathbb{N}^2$
$$
\mathbb{E}\left[\left(\left\{\widehat{\mathbb{E}_{\pi^l}[\varphi(\theta,V_{1:T})] - \mathbb{E}_{\pi^{l-1}}[\varphi(\theta,V_{1:T})]}\right\}^{M_l}-
\{\mathbb{E}_{\pi^l}[\varphi(\theta,V_{1:T})] - \mathbb{E}_{\pi^{l-1}}[\varphi(\theta,V_{1:T})]\}
\right)^2\right] \leq\frac{C\Delta_l^{2\beta}}{M_l}.
$$
\end{prop}

The proof of the following result can be found in Appendix \ref{app:2}.
\begin{prop}\label{prop:main_res1}
Assume Assumptions \ref{ass:1}-\ref{ass:3}. Then for any $\varphi\in\textrm{\emph{Lip}}(\Theta\times\mathbb{R}^T)\cap\mathcal{B}_b(\Theta\times\mathbb{R}^T)$ there exists a $C<\infty$ such that for any $(l,M_l)\in\mathbb{N}^2$
$$
\Bigg|\mathbb{E}\left[\left\{\widehat{\mathbb{E}_{\pi^l}[\varphi(\theta,V_{1:T})] - \mathbb{E}_{\pi^{l-1}}[\varphi(\theta,V_{1:T})]}\right\}^{M_l}-
\{\mathbb{E}_{\pi^l}[\varphi(\theta,V_{1:T})] - \mathbb{E}_{\pi^{l-1}}[\varphi(\theta,V_{1:T})]\}
\right]\Bigg| \leq \frac{C\Delta_l^\beta}{M_l}.
$$
\end{prop}

The implications of these results are as follows. We have that
\begin{equation}\label{eq:mse_bound}
\mathbb{E}\left[\left(\left\{\widehat{\mathbb{E}_{\pi^L}[\varphi(\theta,V_{1:T})]}\right\}^{M_{0:L}} - 
\mathbb{E}_{\pi^L}[\varphi(\theta,V_{1:T})]\right)^2\right] 
\leq C\left(\sum_{l=0}^L \frac{\Delta_l^{2\beta}}{M_l}+
\sum_{l=0}^L\sum_{q=0, q\neq l}^L \frac{\Delta_l^{\beta}}{M_l}\frac{\Delta_q^{\beta}}{M_q}
\right).
\end{equation}
The cost of simulation is of $\mathcal{O}(\sum_{l=0}^L \Delta_l^{-2}M_l)$. 
A bound on the MSE is the sum of the square bias and the upper bound in \eqref{eq:mse_bound}.
The bias (weak error), to our knowledge is not known exactly (see \cite{bayer} for some related results however), but one can bound the square weak error by the strong error (which is $2\beta$ under (A\ref{ass:2})) by Jensen's inequality. Then it is well-known (e.g.~\cite{giles,giles1}) how to choose $L,M_{0:L}$. In the context of Volterra SDEs, the results of 
\cite{richard} suggest that $\beta=H+\tfrac{1}{2}$ in our setting and under the assumptions of that paper. Therefore, for $H\in(0,\tfrac{1}{2})$, one can choose $M_l=\mathcal{O}(\epsilon^{-2}\Delta_l^{\tfrac{2H+3}{2}}\Delta_L^{H-\tfrac{1}{2}})$, where $\epsilon>0$ given and $\Delta_L^{2H+1}=\mathcal{O}(\epsilon^2)$ and thus the MSE is $\mathcal{O}(\epsilon^2)$ for a cost of $\mathcal{O}(\epsilon^{-\tfrac{4}{2H+1}})$. In the case of a single level (i.e.~Algorithm \ref{alg:pmcmc}) the cost to achieve the same MSE is
$\mathcal{O}(\epsilon^{-\tfrac{2(2H+3)}{2H+1}})$ which is typically much higher. In our simulations we use $H=0.4$ which indicates the cost of the multilevel method is almost optimal at $\mathcal{O}(\epsilon^{-\tfrac{20}{9}})$, whereas it is $\mathcal{O}(\epsilon^{-\tfrac{38}{9}})$ for the single level case.


\section{Numerical Results}
\label{sec:num} 

{In the following numerical analysis, the strategy for choosing $L$ and $M_{0:L}$ involves setting a predetermined threshold $\epsilon^2$ for the MSE and subsequently setting $L= \mathcal{O}\left(-2\log(\epsilon)/(2H+1)\right)$ and $M_l=\mathcal{O}(\epsilon^{-2}\Delta_l^{\tfrac{2H+3}{2}}\Delta_L^{H-\tfrac{1}{2}})$, for $l \in \{0,\cdots,L\}$. }

\subsection{Prior Settings}

For the priors, in the context of all simulations that we performed, we consider only $H=0.4$. One can recall that $K(t)=Ct^{H}$ and we set $C=0.7$ in all of our simulations. For the state-space model, all unknown parameters are taken as independent and
on the log scale $\mathcal{N}(0,1)$ random variables. For the SV model again all parameters are independent and on the log scale all parameters, apart from $\rho,r$, are 
$\mathcal{N}(0,1)$. For $\rho$ we assume that, \emph{apriori}:
$$
\log\left(\frac{1+\rho}{1-\rho}\right)\sim\mathcal{N}(0,1).
$$
We take $r\sim\mathcal{N}(0,1)$.

\subsection{Synthetic Data}
\label{subsec:synth_data}
We generate 100 data points from the state-space model (with time-discretized dynamics), when the data are taken as Gaussian with mean as the hidden state and variance $0.8^2$. We also generate 100 data points from the SV model. The discretization is taken at {level $l = 6$} in both cases. For the MCMC moves we use Gaussian random walk proposals with diagonal covariance matrix and tune the MCMC before presenting the results. The samples for single level and multilevel are set using the guidelines in the previous section and we consider levels of discretization up-to {level $l = 5$}.

In the context of the state-space model with synthetic data using single level and multilevel MCMC, we can see the performance of the sampler in Figure \ref{fig:pmmh_SS}, which shows that for the $10^4$ samples presented a generally good mixing of the samplers; note that the targets are not same as one has to correct by importance sampling in the case of multilevel MCMC. The plots for the case of the SV model in this case are also similar, except with extra parameters.

We now investigate the relationship between cost and MSE as conjectured in the previous section. This is considered in Tables \ref{tab:1}-\ref{tab:2} which estimate the rate of decrease of the MSE as the cost increases on the log scale, which should be around $-\tfrac{20}{18}$ for the multilevel method and $-\tfrac{38}{18}$ for the single level method, in terms of each parameter. The results in the Table are slightly better than expected, but indicate that these rates are reasonable.

\begin{table}[h!]
    \centering
    \caption{State-Space Model: Estimated Log Cost against Log MSE. This is for synthetic data.}
    \begin{tabular}{c|c|c}
    Parameter & PMCMC & MLPMCMC \\
    \hline
    $\log(V_0)$ & -1.35 & -1.05 \\
    $\log(\kappa)$ & -1.43 & -1.07\\
    $\log(\lambda)$ &-1.55 & -1.08\\
    $\log(\nu)$ &-1.57 & -1.04
    \end{tabular}
    \label{tab:1}
\end{table}

\begin{table}[h!]
    \centering
    \caption{Stochastic Volatility Model: Estimated Log Cost against Log MSE. This is for synthetic data.}
    \begin{tabular}{c|c|c}
    Parameter & PMCMC & MLPMCMC \\
    \hline
    $\log(V_0)$ & -1.52 & -1.18 \\
    $\log(\tfrac{1+\rho}{1-\rho})$ & -1.64 & -1.07 \\
    $\log(\kappa)$ & -1.45 & -1.10\\
    $\log(\lambda)$ &-1.56 & -1.20\\
    $\log(\nu)$ &-1.56 & -1.26\\
    $r$ &-1.50 & -1.09\\
    \end{tabular}
    \label{tab:2}
\end{table}

\begin{figure}[h!]
    \centering
    \includegraphics[scale=0.478, trim={2.5cm 3cm 2.5cm 3cm},clip]{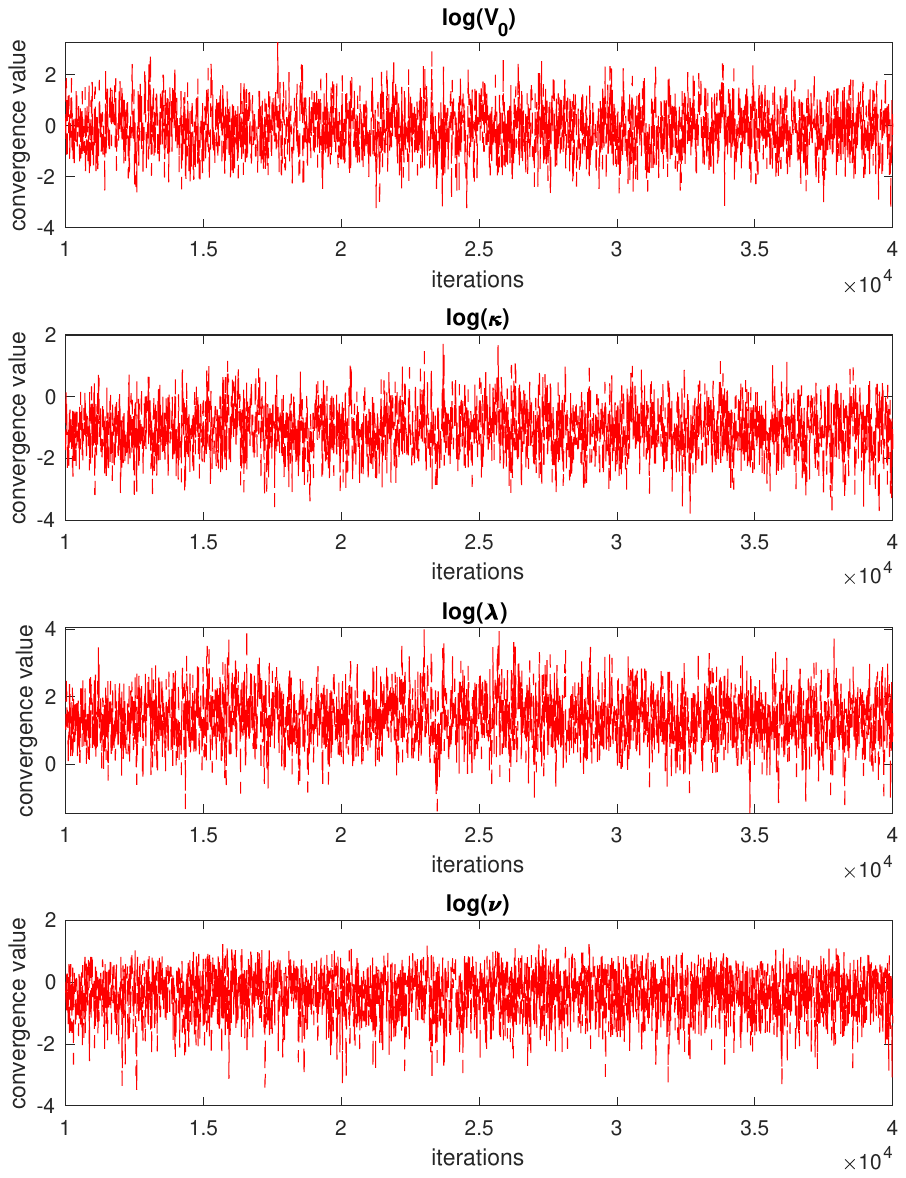}
    \includegraphics[scale=0.478, trim={2.5cm 3cm 2.5cm 3cm},clip]{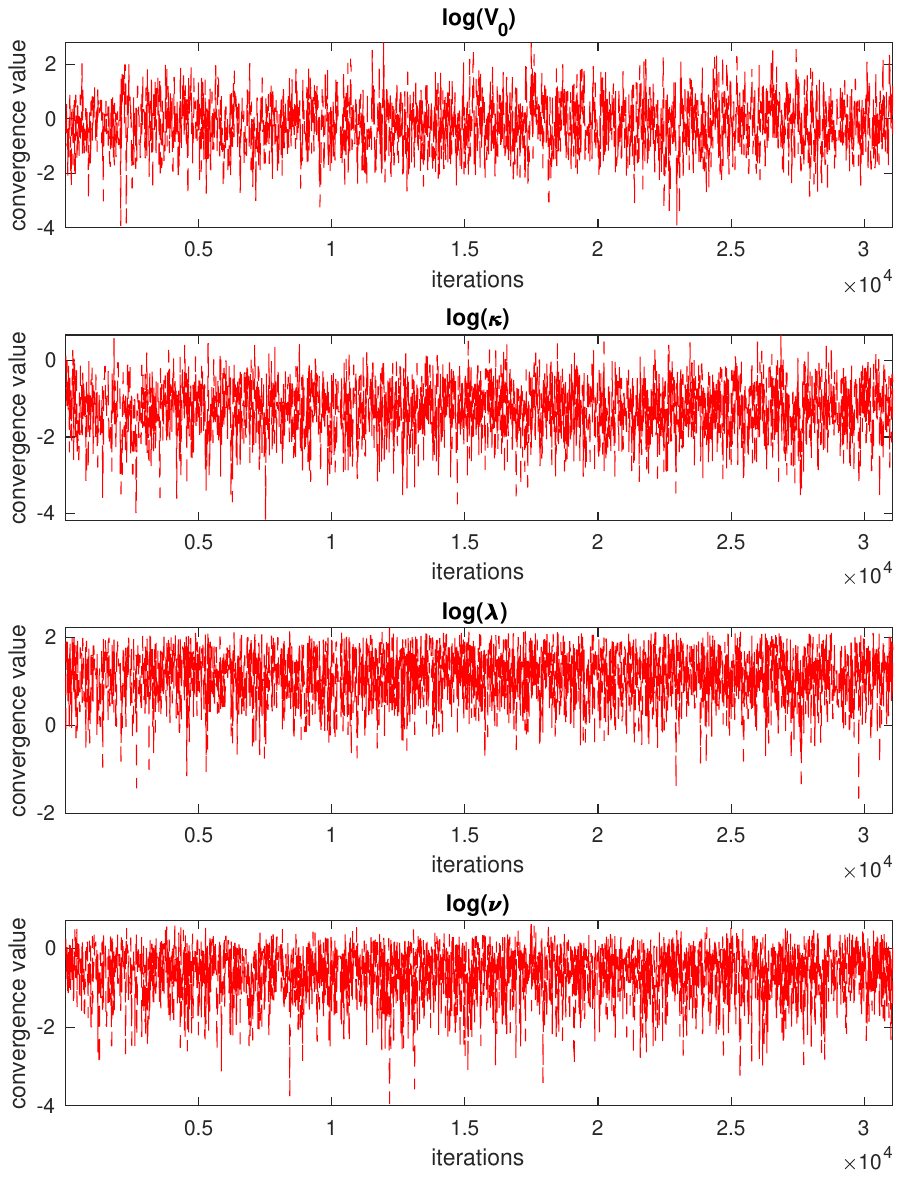}
    \caption{State-Space model for synthetic data: Convergence plot using Particle MCMC (left) and the (marginal) Multilevel Sampler at {level $l = 5$} (right).}
    \label{fig:pmmh_SS}
\end{figure}

\subsection{Real Data}
\label{subsec:real_data}
We consider the daily closing price of the Tadawul All Share Index. The data are recorded from 20 February 2022 to 21 February 2023, which consists of 249 data points and the daily log-returns are plotted on the left of Figure \ref{fig:dat_corr}. We also plot the correlation between the absolute log returns and the lagged returns on the right of Figure \ref{fig:dat_corr}. The plot illustrates the leverage effect, that is, using the absolute log returns as a proxy volatility, that the future volatility is negatively correlated with the log-returns, but the opposite relationship does not hold. We will investigate whether the multilevel method can recover some of the information in the latter plot, by estimating the same quantity from the posterior predictive.

We use our multilevel method, with lowest level as 3 and largest level of 5 and some estimates from the predictive distribution, when compared with the real data can be found in Table \ref{tab:3}. Although the statistics are not exactly the same, it does indicate that the model can represent some features of this data. Turning to Figure \ref{fig:corr_pred} we can see the predictive estimate of the correlation between the absolute log returns and the lagged returns. The structure of this plot does not appear to recover the leverage effect that was prevalent in the real data. At least in the example of this data, the long-term memory of the volatility does not seem to improve on existing studies with more conventional volatility processes (e g.~\cite{jasra_levy}). In Table~\ref{tab:param_est}, we present the parameters mean values obtained using the MLPMCMC method.

\begin{figure}[h!]
    \centering
    \includegraphics[scale=0.55]{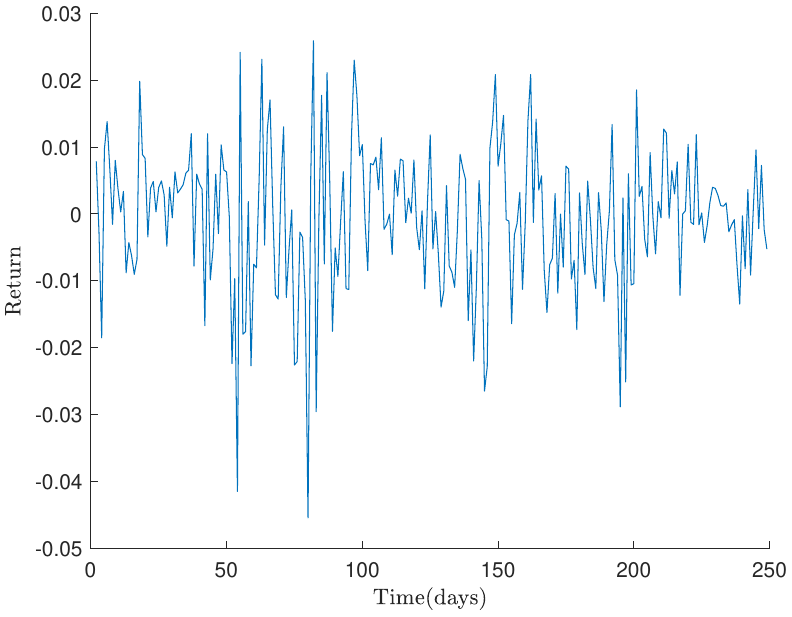}\quad
    \includegraphics[scale=0.55]{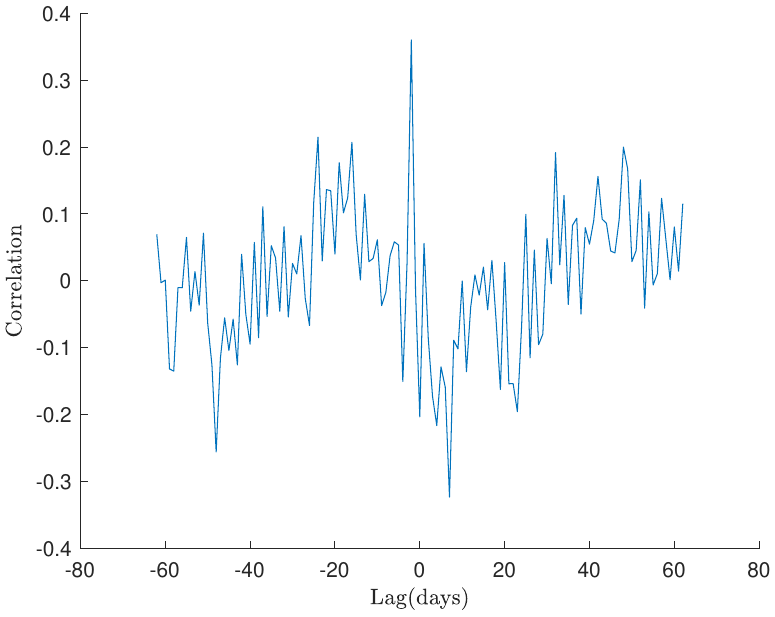}
    \caption{The data and correlation structure. Daily returns from the TASI data (Left), the correlation between the absolute returns $|R_{t_i}|= |\log(Y_{t_i} ) - \log(Y_{t_{i-1}} ) |$ and the lagged returns $R_{t_{i-j}}$ (Right).}
    \label{fig:dat_corr}
\end{figure}

\begin{figure}[h!]
    \centering
    \includegraphics[scale=0.55]{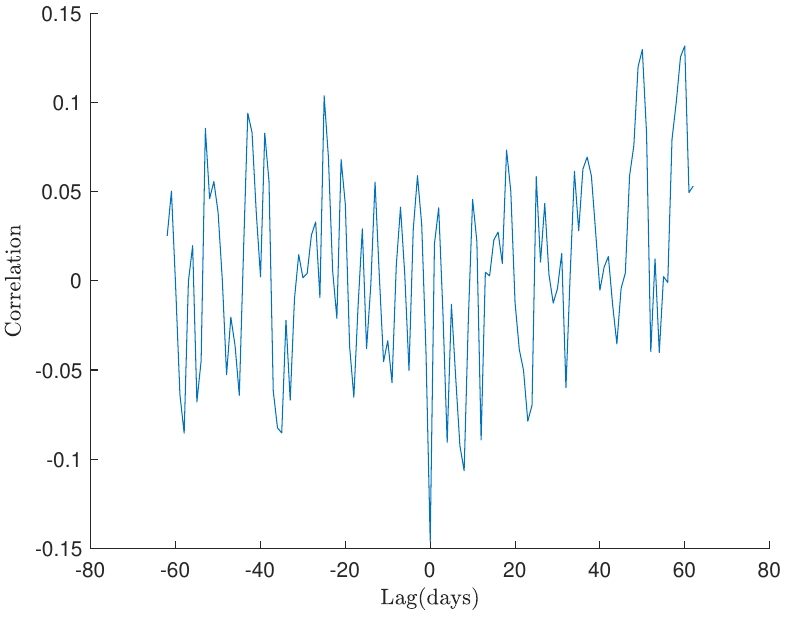}
    \caption{Predictive correlation between absolute returns and lagged returns from the TASI data.}
    \label{fig:corr_pred}
\end{figure}

\begin{table}[h!]
    \centering
    \caption{Summary statistics from the return of the TASI data (249, daily returns). The predictive estimate is a multilevel approximation of the predictive expectation.}
    \begin{tabular}{c|c|c|c|c}
    Return & Mean & Variance & Skewness & Kurtosis \\
\hline
    $Y_t$ & -6.75e-04 & 1.15e-04 & -0.60 & 4.58 \\
    Predictive Estimate & 5.42e-04 & 2.27e-05 & -0.38& 5.39\\
    \end{tabular}
    \label{tab:3}
\end{table}

\begin{table}[h!]
    \centering
    \caption{Parameters mean values obtained via MLPMMH for TASI data (with 249 values of daily returns.)}
    \label{tab:param_est}
    \begin{tabular}{c|c|c|c|c|c|c|c}
    Parameter & $V_0$ & $H$ & $\rho$ & $\kappa$ & $\lambda$ & $\nu$ & r \\
    \hline
     Estimated Values & 0.16447 & 0.4557  & -0.43407 & 0.0017 & 11.0022 & 0.0445 &-5.7093e-04 
    \end{tabular}
\end{table}

\section*{Acknowledgements}
All authors were supported by KAUST baseline funding.

\appendix

\section{Proof of Proposition \ref{prop:main_res}}\label{app:1}

Our proof is given by way of two technical Lemmata. The first of which is a simple discretization result and the second and analogue of \cite[Lemma A.2.]{jasra}. These two results together can be combined with \cite[Proposition A.2.]{jasra} to conclude the result. Throughout the appendix $C$ is a generic finite constant that is independent of $l$. Any dependencies on the model parameters will be clear from the result statements.

In the following Lemma, the expectation is taken w.r.t.~the underlying law of the process $\{V_t\}_{t\in[0,T]}$ where the discretized process $\{V_t^l\}_{t\in\{\Delta_l,2\Delta_l,\dots,T\}}$ shares the same Brownian motion as $\{V_t\}_{t\in[0,T]}$.
\begin{lem}\label{lem:lem1}
Assume Assumptions \ref{ass:1}-\ref{ass:2}. Then for any $(y,y',r)\in\mathsf{Y}^2\times[1,2]$ there exists a $C<\infty$ such that for any $(l,t)\in\mathbb{N}\times\{1,\dots,T\}$
$$
\sup_{\theta\in\Theta}\mathbb{E}\Bigg[\Bigg|\overline{g}\Big(y;y',\Delta_l\sum_{k=0}^{\Delta_l^{-1}-1} \overline{f}(V_{t-1+k\Delta_l}^l),\sum_{k=0}^{\Delta_l^{-1}-1} \overline{h}(V_{t-1+k\Delta_l}^l)\overline{W}_{t-1+(k+1)\Delta_l}\Big)-
$$
$$
\overline{g}\Big(y;y',\int_{t-1}^t \overline{f}(V_s)ds,\int_{t-1}^t \overline{h}(V_s)dW_s\Big)
\Bigg|^r\Bigg]^{1/r} \leq C\Delta_l^{\beta}.
$$
\end{lem}

\begin{proof}
We have the simple decomposition
$$
\mathbb{E}\Bigg[\Bigg|\overline{g}\Big(y;y',\Delta_l\sum_{k=0}^{\Delta_l^{-1}-1} \overline{f}(V_{t-1+k\Delta_l}^l),\sum_{k=0}^{\Delta_l^{-1}-1} \overline{h}(V_{t-1+k\Delta_l}^l)\overline{W}_{t-1+(k+1)\Delta_l}\Big)-
$$
$$
\overline{g}\Big(y;y',\int_{t-1}^t \overline{f}(V_s)ds,\int_{t-1}^t \overline{h}(V_s)dW_s\Big)
\Bigg|^r\Bigg]^{1/r} \leq T_1 + T_2
$$
where
\begin{eqnarray*}
T_1 & = & \mathbb{E}\Bigg[\Bigg|\overline{g}\Big(y;y',\Delta_l\sum_{k=0}^{\Delta_l^{-1}-1} \overline{f}(V_{t-1+k\Delta_l}^l),\sum_{k=0}^{\Delta_l^{-1}-1} \overline{h}(V_{t-1+k\Delta_l}^l)\overline{W}_{t-1+(k+1)\Delta_l}\Big)-\\ & &
\overline{g}\Big(y;y',\Delta_l\sum_{k=0}^{\Delta_l^{-1}-1} \overline{f}(V_{t-1+k\Delta_l}),\sum_{k=0}^{\Delta_l^{-1}-1} \overline{h}(V_{t-1+k\Delta_l})\overline{W}_{t-1+(k+1)\Delta_l}\Big)\Bigg|^r\Bigg]^{1/r}\\
T_2 & = & \mathbb{E}\Bigg[\Bigg|\overline{g}\Big(y;y',\Delta_l\sum_{k=0}^{\Delta_l^{-1}-1} \overline{f}(V_{t-1+k\Delta_l}),\sum_{k=0}^{\Delta_l^{-1}-1} \overline{h}(V_{t-1+k\Delta_l})\overline{W}_{t-1+(k+1)\Delta_l}\Big) - \\
& & \overline{g}\Big(y;y',\int_{t-1}^t \overline{f}(V_s)ds,\int_{t-1}^t \overline{h}(V_s)dW_s\Big)
\Bigg|^r\Bigg]^{1/r}.
\end{eqnarray*}
To conclude the proof, we simply need to bound $T_1$ and $T_2$. As the proofs are very similar, we shall only consider the case for $T_1$ (i.e.~very simple Martingale-reminder type arguments as in \cite{jasra_cont} for instance).

We begin by applying (A\ref{ass:1})-2.~followed by the Minkowski inequality to yield the upper-bound $T_1\leq C(T_3+T_4)$ where
\begin{eqnarray*}
T_3 & = & \mathbb{E}\Bigg[\Bigg|\Delta_l\sum_{k=0}^{\Delta_l^{-1}-1} \overline{f}(V_{t-1+k\Delta_l}^l) - 
\Delta_l\sum_{k=0}^{\Delta_l^{-1}-1} \overline{f}(V_{t-1+k\Delta_l})\Bigg|^r\Bigg]^{1/r} \\
T_4 & = & \mathbb{E}\Bigg[\Bigg|\sum_{k=0}^{\Delta_l^{-1}-1} \{\overline{h}(V_{t-1+k\Delta_l}^l)-\overline{h}(V_{t-1+k\Delta_l})\}\overline{W}_{t-1+(k+1)\Delta_l}\Bigg|^r\Bigg]^{1/r}.
\end{eqnarray*}
For $T_3$ using the Minkowski inequality multiple times, one has
$$
T_3 \leq \Delta_l\sum_{k=0}^{\Delta_l^{-1}-1}\mathbb{E}\Bigg[\Bigg|\overline{f}(V_{t-1+k\Delta_l}^l) -
\overline{f}(V_{t-1+k\Delta_l})\Bigg|^r\Bigg]^{1/r}.
$$
Then using (A\ref{ass:1}) 1.~followed by (A\ref{ass:2}) gives the upper-bound
$$
T_3 \leq C\Delta_l^{\beta}.
$$
For $T_4$ one can show that (w.r.t.~an appropriate filtration) that 
$$\sum_{k=0}^{\Delta_l^{-1}-1} \{\overline{h}(V_{t-1+k\Delta_l}^l)-\overline{h}(V_{t-1+k\Delta_l})\}\overline{W}_{t-1+(k+1)\Delta_l}$$ 
is a martingale and hence, by application of the Burkholder-Gundy-Davis inequality, we have the upper-bound
$$
T_4 \leq C\Delta_l^{1/2}
\mathbb{E}\Bigg[\Bigg|\sum_{k=0}^{\Delta_l^{-1}-1} \{\overline{h}(V_{t-1+k\Delta_l}^l)-\overline{h}(V_{t-1+k\Delta_l})\}^2\Bigg|^{\tfrac{r}{2}}\Bigg]^{1/r}.
$$
Then by using the Minkowski inequality multiple times along with (A\ref{ass:1}) 1.~gives
$$
T_4 \leq  C\Delta_l^{1/2}\Bigg(
\sum_{k=0}^{\Delta_l^{-1}-1}\mathbb{E}\Bigg[\Bigg|V_{t-1+k\Delta_l}^l)-V_{t-1+k\Delta_l}\Bigg|^{r}\Bigg]^{2/r}
\Bigg)^{1/2}.
$$
Application of (A\ref{ass:2}) gives the upper-bound
$$
T_4 \leq C\Delta_l^{\beta}
$$
and from here the proof can be concluded.
\end{proof}

\begin{lem}\label{lem:lem2}
Assume Assumptions \ref{ass:1}-\ref{ass:2}. Then for any $(r,\varphi)\in\{1,2\}\times\textrm{\emph{Lip}}(\Theta\times\mathbb{R}^T)\cap\mathcal{B}_b(\Theta\times\mathbb{R}^T)$ there exists a $C<\infty$ such that for any $l\in\mathbb{N}$
$$
\mathbb{E}_{\tilde{\pi}^{l,l-1}}[|\varphi(\theta,V_{1:T}^l)H_{1,l}(\overline{W}_{\Delta_l:T}^l,\overline{W}_{\Delta_{l-1}:T}^{l-1})-\varphi(\theta,V_{1:T}^{l-1})H_{2,l}(\overline{W}_{\Delta_l:T}^l,\overline{W}_{\Delta_{l-1}:T}^{l-1})|^r]^{3-r} \leq C\Delta_l^{2\beta}.
$$
\end{lem}

\begin{proof}
We give the proof for $r=1=T$ only. Extension to the general case can be conducted via induction on $T$ and to the case $r=2$ is a similar proof to the case $r=1$.

Then we have, via (A\ref{ass:1}) 3.~that
$$
\mathbb{E}_{\tilde{\pi}^{l,l-1}}[|\varphi(\theta,V_{1}^l)H_{1,l}(\overline{W}_{\Delta_l:1}^l,\overline{W}_{\Delta_{l-1}:1}^{l-1})-\varphi(\theta,V_{1}^{l-1})H_{2,l}(\overline{W}_{\Delta_l:1}^l,\overline{W}_{\Delta_{l-1}:1}^{l-1})|]
\leq 
$$
$$
C\int_{\Theta}\mathbb{E}_{\tilde{Q}^{l,l-1}}\Bigg[\Bigg|\varphi(\theta,V_{1}^l)\overline{g}\Big(y_1;y_0,\Delta_l\sum_{k=0}^{\Delta_l^{-1}-1} \overline{f}(V_{k\Delta_l}^l),\sum_{k=0}^{\Delta_l^{-1}-1} \overline{h}(V_{k\Delta_l}^l)\overline{W}_{(k+1)\Delta_l}\Big)
$$
$$
-\varphi(\theta,V_{1}^{l-1})
\overline{g}\Big(y_1;y_0,\Delta_{l-1}\sum_{k=0}^{\Delta_{l-1}^{-1}-1} \overline{f}(V_{k\Delta_{l-1}}^{l-1}),\sum_{k=0}^{\Delta_{l-1}^{-1}-1} \overline{h}(V_{k\Delta_{l-1}}^{l-1})\overline{W}_{(k+1)\Delta_{l-1}}\Big)
\Bigg|\Bigg]\pi(\theta)d\theta.
$$
The R.H.S.~is upper-bounded by $C(T_1+T_2)$ where
\begin{eqnarray*}
T_1 & = & \int_{\Theta}\mathbb{E}_{\tilde{Q}^{l,l-1}}\Big[\Big|\varphi(\theta,V_{1}^l)-\varphi(\theta,V_{1}^{l-1})\Big|\Big]\pi(\theta)d\theta\\
T_2 & = & \int_{\Theta}\mathbb{E}_{\tilde{Q}^{l,l-1}}\Bigg[\Bigg|\overline{g}\Big(y_1;y_0,\Delta_l\sum_{k=0}^{\Delta_l^{-1}-1} \overline{f}(V_{k\Delta_l}^l),\sum_{k=0}^{\Delta_l^{-1}-1} \overline{h}(V_{k\Delta_l}^l)\overline{W}_{(k+1)\Delta_l}\Big)
-\\ & & \overline{g}\Big(y_1;y_0,\Delta_{l-1}\sum_{k=0}^{\Delta_{l-1}^{-1}-1} \overline{f}(V_{k\Delta_{l-1}}^{l-1}),\sum_{k=0}^{\Delta_{l-1}^{-1}-1} \overline{h}(V_{k\Delta_{l-1}}^{l-1})\overline{W}_{(k+1)\Delta_{l-1}}\Big)\Bigg|\Bigg]\pi(\theta)d\theta.
\end{eqnarray*}
The term $T_1$ can be dealt with via $\varphi\in\textrm{Lip}(\Theta\times\mathbb{R}^T)$ and (A\ref{ass:2}) and the $T_2$ can be handled by using Lemma \ref{lem:lem1} which concludes the proof.
\end{proof}

\section{Proof of Proposition \ref{prop:main_res1}}\label{app:2}

\begin{proof}
The proof is similar to \cite[Proposition A.1.]{jasra_cont} and is modified marginally to the problem in this article.
By using (e.g.) \cite[Lemma C.5.]{mlpf} we have the decomposition
$$
\left\{\widehat{\mathbb{E}_{\pi^l}[\varphi(\theta,V_{1:T})] - \mathbb{E}_{\pi^{l-1}}[\varphi(\theta,V_{1:T})]}\right\}^{M_l}-
\{\mathbb{E}_{\pi^l}[\varphi(\theta,V_{1:T})] - \mathbb{E}_{\pi^{l-1}}[\varphi(\theta,V_{1:T})]\} = 
$$
$$
\frac{\{a^{M_l}-b^{M_l}-(a-b)\}}{A^{M_l}} - \frac{b^{M_l}\{A^{M_l}-B^{M_l}-(A-B)\}}{A^{M_l}B^{M_l}} + \frac{1}{A^{M_l}A}[A-A^{M_l}][a-b] - \frac{1}{A^{M_l}B^{M_l}}[b^{M_l}-b][A-B]  
$$
$$
+\frac{b}{AB^{M_l}B}[B^{M_l}-B][A-B] + 
\frac{b}{AA^{M_l}B^{M_l}}[A^{M_l}-A][A-B]
$$
where
\begin{eqnarray*}
a^{M_l} & = & \frac{1}{M+1}\sum_{k=0}^{M_l}\varphi(\theta_k^l,v_{k,1:T}^l(l))H_{1,l}(\theta_k^l,\overline{w}_{k,\Delta_l:T}^l(l),\overline{w}_{k,\Delta_{l-1}:T}^{l-1}(l)) \\
A^{M_l} & = & \frac{1}{M+1}\sum_{k=0}^{M_l}H_{1,l}(\theta_k^l,\overline{w}_{k,\Delta_l:T}^l(l),\overline{w}_{k,\Delta_{l-1}:T}^{l-1}(l)) \\
b^{M_l} & = & \frac{1}{M+1}\sum_{k=0}^{M_l}\varphi(\theta_k^l,v_{k,1:T}^{l-1}(l))H_{2,l}(\theta_k^l,\overline{w}_{k,\Delta_l:T}^l(l),\overline{w}_{k,\Delta_{l-1}:T}^{l-1}(l)) \\
B^{M_l} & = & \frac{1}{M+1}\sum_{k=0}^{M_l}H_{2,l}(\theta_k^l,\overline{w}_{k,\Delta_l:T}^l(l),\overline{w}_{k,\Delta_{l-1}:T}^{l-1}(l))
\end{eqnarray*}
and 
\begin{eqnarray*}
a & = & \mathbb{E}_{\tilde{\pi}^{l,l-1}}[\varphi(\theta,V_{1:T}^l)H_{1,l}(\overline{W}_{\Delta_l:T}^l,\overline{W}_{\Delta_{l-1}:T}^{l-1})] \\
A & = & \mathbb{E}_{\tilde{\pi}^{l,l-1}}[H_{1,l}(\overline{W}_{\Delta_l:T}^l,\overline{W}_{\Delta_{l-1}:T}^{l-1})]\\
b & = & \mathbb{E}_{\tilde{\pi}^{l,l-1}}[\varphi(\theta,V_{1:T}^{l-1})H_{2,l}(\overline{W}_{\Delta_l:T}^l,\overline{W}_{\Delta_{l-1}:T}^{l-1})] \\
B &=& \mathbb{E}_{\tilde{\pi}^{l,l-1}}[H_{2,l}(\overline{W}_{\Delta_l:T}^l,\overline{W}_{\Delta_{l-1}:T}^{l-1})].
\end{eqnarray*}
As the chain is started in stationarity, one has
$$
\left\{\widehat{\mathbb{E}_{\pi^l}[\varphi(\theta,V_{1:T})] - \mathbb{E}_{\pi^{l-1}}[\varphi(\theta,V_{1:T})]}\right\}^{M_l}-
\{\mathbb{E}_{\pi^l}[\varphi(\theta,V_{1:T})] - \mathbb{E}_{\pi^{l-1}}[\varphi(\theta,V_{1:T})]\} = 
$$
$$
\mathbb{E}\left[\left(\frac{1}{A^{M_l}}-\frac{1}{A}\right)\{a^{M_l}-b^{M_l}-(a-b)\}\right] - 
\mathbb{E}\left[\left(
\frac{b^{M_l}}{A^{M_l}B^{M_l}} - \frac{b}{AB}
\right)\{A^{M_l}-B^{M_l}-(A-B)\}\right] + 
$$
$$
[a-b]\mathbb{E}\left[\left(\frac{1}{A^{M_l}A}-\frac{1}{A^2}\right)[A-A^{M_l}]\right] - 
[A-B]\mathbb{E}\left[\left(\frac{1}{A^{M_l}B^{M_l}} -
\frac{1}{AB}
\right)[b^{M_l}-b]\right] + 
$$
$$
[A-B]\mathbb{E}\left[\left(\frac{b}{AB^{M_l}B}
- \frac{b}{AB^2}
\right)[B^{M_l}-B]\right] + 
[A-B]\mathbb{E}\left[\left(
\frac{b}{AA^{M_l}B^{M_l}}-
\frac{b}{A^2B}
\right)[A^{M_l}-A]\right].
$$
The control of many of these terms is fairly repetitive so,  we consider only the first and third terms on the R.H.S.~of the displayed equation. For the first term, one can apply Cauchy-Schwarz and \cite[Proposition A.1.]{jasra} to yield
$$
\mathbb{E}\left[\left(\frac{1}{A^{M_l}}-\frac{1}{A}\right)\{a^{M_l}-b^{M_l}-(a-b)\}\right] \leq 
\mathbb{E}\left[\left(\frac{1}{A^{M_l}}-\frac{1}{A}\right)^2\right]^{1/2}\frac{C\Delta_l^{\beta}}{M_l^{1/2}}
$$
Using standard results for the convergence of Markov chains (e.g.~\cite[Proposition A.1.]{jasra}) then we arrive
at
$$
\mathbb{E}\left[\left(\frac{1}{A^{M_l}}-\frac{1}{A}\right)\{a^{M_l}-b^{M_l}-(a-b)\}\right] \leq \frac{C\Delta_l^{\beta}}{M_l}.
$$
For the third term on the R.H.S.~application of Lemma \ref{lem:lem2} along with Cauchy Schwarz gives
$$
[a-b]\mathbb{E}\left[\left(\frac{1}{A^{M_l}A}-\frac{1}{A^2}\right)[A-A^{M_l}]\right] \leq C\Delta_l^{\beta}
\mathbb{E}\left[\left(\frac{1}{A^{M_l}A}-\frac{1}{A^2}\right)^2\right]^{1/2}
\mathbb{E}\left[[A-A^{M_l}]^2\right] ^{1/2}
$$
again using standard results for the convergence of Markov chains
$$
[a-b]\mathbb{E}\left[\left(\frac{1}{A^{M_l}A}-\frac{1}{A^2}\right)[A-A^{M_l}]\right] \leq \frac{C\Delta_l^{\beta}}{M_l}
$$
and this concludes the proof.
\end{proof}

\end{document}